\documentclass[11pt,letterpaper]{article}

\usepackage[margin=1in]{geometry}
\usepackage[T1]{fontenc}
\usepackage{lmodern}
\usepackage{amsmath,amssymb,amsthm,mathtools}
\usepackage{booktabs}
\usepackage{enumitem}
\usepackage{microtype}
\usepackage{xcolor}
\usepackage{authblk}
\usepackage[
  hidelinks,
  pdfauthor={Chenghua Liu and Boning Meng},
  pdftitle={Bounded Relative Boundary Implies Narrow DNF Approximation}
]{hyperref}

\newtheorem{theorem}{Theorem}[section]
\newtheorem{lemma}[theorem]{Lemma}

\newtheorem{conjecture}[theorem]{Conjecture}
\theoremstyle{definition}

\theoremstyle{remark}
\newtheorem{remark}[theorem]{Remark}

\newcommand{\bits}{\{0,1\}}
\newcommand{\E}{\mathbb E}
\renewcommand{\Pr}{\mathbb P}
\newcommand{\Var}{\operatorname{Var}}
\newcommand{\one}{\mathbf 1}
\newcommand{\sens}{\operatorname{sens}}
\newcommand{\cA}{\mathcal A}
\newcommand{\cB}{\mathcal B}
\newcommand{\cE}{\mathcal E}
\newcommand{\cF}{\mathcal F}
\newcommand{\cJ}{\mathcal J}
\newcommand{\cS}{\mathcal S}

\setlist[itemize]{leftmargin=2em,itemsep=0.2em,topsep=0.35em}
\setlist[enumerate]{leftmargin=2.25em,itemsep=0.2em,topsep=0.35em}
\allowdisplaybreaks

\title{Bounded Relative Boundary Implies\\[2.4pt]
Narrow DNF Approximation}
\author[1]{Chenghua Liu}
\author[2]{Boning Meng}
\affil[1]{Institute of Software, Chinese Academy of Sciences, Beijing, China}
\affil[2]{University of Regensburg, Regensburg, Germany}
\affil[ ]{\href{mailto:liuch.russell@gmail.com}{\texttt{liuch.russell@gmail.com}},
  \href{mailto:mengboning2013@gmail.com}{\texttt{mengboning2013@gmail.com}}}
\date{}

\usepackage{tikz}
\definecolor{pinkfoldpink}{HTML}{B74B7D}
\newcommand{\pinkfoldmark}{%
  \raisebox{-0.15em}{%
    \begin{tikzpicture}[x=0.0375em,y=-0.0375em]
      \path[use as bounding box] (0,0) rectangle (24,24);
      \draw[pinkfoldpink,line width=0.06375em,line join=miter,line cap=butt]
        (12,2) -- (20.660,7) -- (20.660,17) -- (12,22) --
        (3.340,17) -- (3.340,7) -- cycle
        (12,12) -- (20.660,7)
        (12,12) -- (3.340,7)
        (12,12) -- (12,22);
    \end{tikzpicture}%
  }%
}

\begin{document}

\pagenumbering{roman}
\maketitle
\thispagestyle{empty}
\thispdfpagelabel{Title}

\begin{abstract}
  \normalsize
  We prove Friedgut's conjecture (\mbox{\emph{J. Amer. Math. Soc.} 12 (1999)})
  that every increasing family in the $p$-biased discrete cube with bounded
  relative boundary can be approximated arbitrarily well by one whose minimal
  elements have size bounded independently of the dimension and the bias.
  For \mbox{$0<p\le1/2$}, every increasing Boolean function with total
  resampling influence at most $K$ is $\varepsilon$-close under $\mu_p^n$
  to a monotone DNF of width $\exp(O((K+1)^2/\varepsilon^2))$, from which
  the full conjecture follows for all $p\in(0,1)$.
  We build on Hatami's pseudo-junta theorem
  (\mbox{\emph{Ann. of Math.} 176 (2012)}), whose approximator need not be
  increasing. 
  Our randomized shifts produce an increasing approximator that is measurable with respect to a controlled refinement of Hatami's adaptive partition.
  Retaining only short positive certificates yields the desired
  narrow DNF.
\end{abstract}

\clearpage
\pagenumbering{arabic}

\section{Introduction}
\label{sec:introduction}

The analysis of Boolean functions studies how local behavior on the discrete
cube constrains global structure.  Foundational results connect local
sensitivity to sharp thresholds and low-complexity approximation
\cite{KahnKalaiLinial1988,FriedgutKalai1996,Friedgut1998}.  A central theme is
that a function that rarely changes under local perturbations should admit a
simple structural description.

Let $\bits^n$ denote the discrete cube and write $[n]=\{1,\ldots,n\}$.  For
$p\in(0,1)$, let $\mu_p$ be the Bernoulli measure on $\bits$ with
$\mu_p(1)=p$, and let
$\mu_p^n$ be its product measure on $\bits^n$.  We write $\Pr_p$, $\E_p$,
and $\Var_p$ for probability, expectation, and variance under $\mu_p^n$, using
the same notation for its coordinate marginals.  The Hamming weight of
$x\in\bits^n$ is
$|x|=\sum_{i=1}^n x_i$.  A Boolean function $f:\bits^n\to\bits$ is
\emph{increasing} if $x\le y$ coordinatewise implies $f(x)\le f(y)$.  If
$x^{\oplus i}$ is obtained by flipping coordinate $i$, define
\[
       \sens_f(x)=\bigl|\{i\in[n]:f(x)\ne f(x^{\oplus i})\}\bigr|.
\]
For $X\sim\mu_p^n$, let $X^{(i)}$ be obtained by independently resampling
coordinate $i$.  We use the total resampling influence
\[
       I_p(f)=\sum_{i=1}^n\Pr_p\bigl[f(X)\ne f(X^{(i)})\bigr].
\]
This is Hatami's resampling convention \cite[eq.~(9)]{Hatami2012}.
In the $p$-biased orthonormal Fourier basis,
$I_p(f)=2\sum_{S\subseteq[n]}|S|\,|\widehat f(S)|^2$.

A \emph{$k$-junta} is a Boolean function that depends on at most $k$
coordinates.  When $p$ is bounded away from zero and one, Friedgut's junta
theorem states that bounded total influence implies approximation by a
$k$-junta for some $k$ independent of $n$ \cite{Friedgut1998}.  Such a bound
on $k$ cannot hold uniformly as $p$ tends to zero.  At $p=1/n$, for example, the
function $\operatorname{OR}_n(x)=x_1\vee\cdots\vee x_n$ satisfies
$I_{1/n}(\operatorname{OR}_n)=2(1-1/n)^n=O(1)$ but remains a constant distance
under $\mu_{1/n}^n$ from every $k$-junta with fixed $k$.  Yet one can certify that
$\operatorname{OR}_n(x)=1$ by revealing a single coordinate equal to one.
Sparse random graphs exhibit the same distinction.  A triangle may occur
anywhere in the graph, but its three edges certify its presence.
Such local witnesses play a central role in threshold
phenomena for random graphs and random constraint-satisfaction problems
\cite{FriedgutKalai1996,Friedgut1999}.

For an increasing Boolean function $f$, the Margulis--Russo identity
\cite{Margulis1974,Russo1981} gives
\begin{equation}
\label{eq:intro-margulis-russo}
  I_p(f)=2p(1-p)\E_p[\sens_f(X)]
  =2p(1-p)\frac{d}{dp}\E_p[f].
\end{equation}
Thus Friedgut's quantity
\[
  p\E_p[\sens_f(X)]
  =\frac{I_p(f)}{2(1-p)}
  =\frac{d}{d\log p}\E_p[f]
\]
measures the derivative of the acceptance probability $\E_p[f]=\Pr_p[f=1]$
on the multiplicative rather than additive scale of $p$.  We therefore refer
to $p\E_p[\sens_f(X)]$ as the relative boundary.

Friedgut's conjecture asks whether every increasing function with bounded
relative boundary can be approximated by a function with short positive
witnesses, as in the example of $\operatorname{OR}_n$.  For
$T\subseteq[n]$, a \emph{positive term} is the conjunction
$\bigwedge_{i\in T}x_i$.  A \emph{monotone DNF} is a disjunction of positive
terms, and its width is the largest $|T|$ among them.  Each term is a local
positive witness because setting its coordinates to one forces acceptance.
We use the usual conventions that the empty conjunction is one, the empty
disjunction is zero, and both constant functions have width zero.
For a statement $P$, write $\one[P]$ for its indicator.  For
$\cB\subseteq\bits^n$, write $\one_{\cB}(x)=\one[x\in\cB]$, and call $\cB$
increasing when $\one_{\cB}$ is increasing.  An element $z\in\cB$ is
minimal if no distinct $y\in\cB$ satisfies $y\le z$.  If $\cB$ is increasing,
the terms $\bigwedge_{i:z_i=1}x_i$ indexed by its minimal elements form the
canonical monotone DNF of $\one_{\cB}$.  Thus $\one_{\cB}$ has a monotone DNF
of width at most $k$ if and only if every minimal element has Hamming weight
at most $k$.  For families
$\cA,\cB\subseteq\bits^n$, write $\cA\mathbin\triangle\cB$ for their symmetric
difference.

\begin{conjecture}[Friedgut's conjecture]
\label{conj:friedgut}
For every $c>0$ and $0<\varepsilon<1$, there is
$k=k(c,\varepsilon)<\infty$ such that the following holds for every
$p\in(0,1)$ and $n\ge1$.  If $\cA\subseteq\bits^n$ is increasing and
\[
       p\E_p[\sens_{\one_{\cA}}(X)]\le c,
\]
then there is an increasing $\cB\subseteq\bits^n$ satisfying
$\mu_p^n(\cA\mathbin\triangle\cB)\le\varepsilon$ whose minimal elements all
have Hamming weight at most $k$.
\end{conjecture}

Our main theorem gives an explicit width bound in the low-bias regime.

\begin{theorem}[Narrow-DNF approximation]
\label{thm:main}
There is an absolute constant $C>0$ such that the following holds.  Let
$n\ge1$, $K>0$, $0<p\le1/2$, and $0<\varepsilon<1$.  If an increasing
function $f:\bits^n\to\bits$ satisfies $I_p(f)\le K$, then there is a
monotone DNF $g$ of width at most
$\exp\!\left(C(K+1)^2/\varepsilon^2\right)$ such that
\[
       \Pr_p[f(X)\ne g(X)]\le\varepsilon.
\]
\end{theorem}

Together with the high-bias reduction, Theorem~\ref{thm:main} proves
Friedgut's conjecture \cite[Conjecture~1.5]{Friedgut1999} for all $p\in(0,1)$.
Our result shows that narrow monotone DNFs,
rather than ordinary juntas, provide the appropriate dimension-free
approximating class for sparse product spaces.

\paragraph{The gap after Hatami.}
Hatami's structure theorem replaces an ordinary junta by a
\emph{pseudo-junta} \cite{Hatami2012}.  For $S\subseteq[n]$, write $x_S$ for
the restriction of $x$ to $S$.  Hatami supplies a family
$\cJ=(J_S)_{S\subseteq[n]}$ of local activation functions
$J_S:\bits^S\to\bits$ and an approximator $h$ of $f$.  On input $x$, the
activations determine the information
\begin{equation}
\label{eq:intro-adaptive-map}
  A_{\cJ}(x)=\bigcup_{S:J_S(x_S)=1}S,
  \qquad
  \Phi_{\cJ}(x)=\bigl(A_{\cJ}(x),x_{A_{\cJ}(x)}\bigr).
\end{equation}
Thus $\Phi_{\cJ}(x)$ records the coordinates activated at $x$ together with
their values.  We call two inputs $x,y\in\bits^n$ members of the same
\emph{adaptive atom} if $\Phi_{\cJ}(x)=\Phi_{\cJ}(y)$, and these atoms form the
adaptive partition of $\bits^n$.  We identify an atom with the corresponding
value of $\Phi_{\cJ}$ when convenient.  The value of $h$ depends only on the
information recorded by $\Phi_{\cJ}$.  Equivalently, there is a readout map
$\Gamma$ such that $h(x)=\Gamma(\Phi_{\cJ}(x))$, so $h$ is constant on every
adaptive atom.  We refer to this property as measurability with respect to the
adaptive partition.  For $X\sim\mu_p^n$, Hatami's representation satisfies
$\Pr_p[f(X)\ne h(X)]=O(\varepsilon)$ and
$\E_p|A_{\cJ}(X)|=O_{K,\varepsilon}(1)$.  When $p\le1/2$, the activations in
Hatami's construction are increasing.  The approximator $h$, however, need
not be increasing, even when the target $f$ is.  For a general family of
increasing activations, the best increasing approximator measurable on the
original partition can have larger error than the best unrestricted
approximator on that partition (Remark~\ref{rem:adaptive-obstruction}).

\paragraph{Proof architecture.}
Starting from $f$, we first obtain Hatami's approximator $h$, then construct an
increasing adaptive approximator $u$, and finally extract a monotone DNF $g$.
Our proof proceeds through the following sequence of transformations.
\[
  f
  \xrightarrow[\text{Hatami}]{\text{adaptive structure}}
  h
  \xrightarrow[\text{this paper}]{\text{controlled monotonicization}}
  u
  \xrightarrow[\text{this paper}]{\text{positive certificates}}
  g.
\]
We extract and track two quantitative bounds from Hatami's activation family.
Every nonzero activation examines at most
$d=O((K+1)/\varepsilon)$ coordinates.  We also
introduce the \emph{multiplicity-counted activation load}
\begin{equation}
\label{eq:intro-load}
  \Lambda_p(\cJ)
  =\sum_{S\subseteq[n]}|S|\Pr_p[J_S(X_S)=1]
\end{equation}
and show directly from Hatami's construction that it is at most
$\exp(O((K+1)^2/\varepsilon^2))$.  The load bounds the expected size of the
active union by counting overlapping sets with multiplicity.  Each summand
depends on at most $d$ coordinates, so we can bound the effect of random
forcing on each summand separately.

For the second arrow, we process the coordinates one at a time.  At each
coordinate, after fixing all other coordinates, we inspect the pair of values
of the current approximator at its zero and one endpoints.  The only pair that
violates monotonicity is $10$.  An upper shift changes this pair to $11$, while
a lower shift changes it to $00$.  We choose the upper shift with probability
$1-p$ and the lower shift with probability $p$.  The randomness is only in the
choice between these two deterministic transformations; the input itself is
not resampled.  Since the zero and one endpoints have $\mu_p$-weights $1-p$
and $p$, respectively, this choice ensures that the expected error against
every increasing target does not increase.  We make the choices independently
across coordinates, and every realization after all coordinates have been
processed is increasing.

The difficult part is to retain controlled adaptive measurability.  We show that
the function produced by an upper shift remains measurable with respect to the
old adaptive partition.  Because the activations are increasing, the atom at
an upper endpoint determines the atom at the corresponding lower endpoint.  A
lower shift requires more information.  Raising the coordinate from zero to
one may create new activations, but the lower atom does not record the values
of the newly activated coordinates.  We encode this missing information by
forcing the shifted coordinate to one inside every activation.  Let $D$ be the
set of coordinates on which lower shifts were chosen.  For $z\in\bits^S$, let
$z^{D\cap S\leftarrow1}$ be obtained
by setting the coordinates in $D\cap S$ to one.  We define the family of
forced activations $\cJ^D=(J_S^D)_{S\subseteq[n]}$ by
$J_S^D(z)=J_S(z^{D\cap S\leftarrow1})$
and show that the final shifted function is measurable with respect to these
activations.
Since each coordinate belongs to $D$ with probability $p$, forcing a random
$D$ changes the effective bias from $p$ to $p'=2p-p^2$.  For each activation
involving at most $d$ coordinates, the likelihood ratio between $\mu_{p'}$ and
$\mu_p$ is at most $2^d$.  Writing $\E_D$ for expectation over the random shift
choices, we obtain
\[
       \E_D\Lambda_p(\cJ^D)\le2^d\Lambda_p(\cJ).
\]
We then use a weighted averaging argument to select one realization with both
small error and small active-set cost.

For the last arrow, we use monotonicity of both $u$ and the forced activations
$\cJ^D$.  Fix $x$ with $u(x)=1$ and set every coordinate outside
$A_{\cJ^D}(x)$ to zero.  Every set active at $x$ remains active because all
of its coordinates are unchanged.  No new set becomes active when coordinates
are lowered, since the activations are increasing.  The adaptive atom, and
hence the value of $u$, is unchanged.  It follows that the coordinates in
$A_{\cJ^D}(x)$ at which $x$ has
value one form an exact positive certificate for $u$, meaning that every input
that is one on these coordinates is accepted by $u$.  We keep only certificates
of size at most a cutoff $w$ to obtain a width-$w$ monotone DNF, and we use
Markov's inequality to bound the discarded mass.

The quadratic exponent in the width bound comes from Hatami's Fourier
truncation, which gives $L=\exp(O((K+1)^2/\varepsilon^2))$.  For activations
of arity $d=O((K+1)/\varepsilon)$, shifting and truncating the certificates
give a width bound of $O(2^d L/\varepsilon)$.

\paragraph{Relation to prior work.}
Friedgut and Kalai established sharp thresholds on the additive $p$-scale for
families invariant under a transitive permutation group
\cite{FriedgutKalai1996}.  Friedgut's junta theorem describes the dense regime
\cite{Friedgut1998}.  His later work used graph and hypergraph symmetries to
show that slowly changing acceptance probabilities in sparse spaces are
governed by bounded local witnesses \cite{Friedgut1999}.  Bourgain's appendix
found a bounded local conditioning that raises the acceptance probability, but
not an approximation by bounded minimal elements
\cite[Appendix]{Friedgut1999}.  Hatami obtained pseudo-junta approximations
without symmetry assumptions and observed that, for increasing targets, it
remained open whether the approximating family could be generated by small
minimal elements \cite[Section~6]{Hatami2012}.  Keller and Lifshitz work
on the uniform cube, in their influence normalization, under the
near-isoperimetric hypothesis
\[
       I(f)\le 2\mu\bigl(\log_2(1/\mu)+M\bigr),\qquad \mu=\E f,
\]
and obtain an $\varepsilon\mu$-approximation by a DNF of size
$2^{2^{O(M/\varepsilon)}}$ \cite{KellerLifshitz2018}.  Their
near-isoperimetric hypothesis and relative-error objective differ from our
bounded-relative-boundary hypothesis and additive-error conclusion; our result
controls width rather than DNF size.  Dinur, Filmus, and Harsha obtained
sparse-DNF structure under the additional assumption that the target is close
to a low-degree polynomial
\cite{DinurFilmusHarsha2024}.  Coordinate shifting is classical and has also
been used for monotone DNF approximators \cite{BlaisHastadServedioTan2014}.
Our contribution is a bias-matched monotonicization of Hatami's adaptive
structure: the resulting function is exactly measurable with respect to a
controlled forced refinement whose expected information cost remains
dimension-free.

\section{Hatami's pseudo-juntas}
\label{sec:pseudo-juntas}

We continue with the notation from the introduction.  For $i\in[n]$ and
$b\in\bits$, write $x^{i\leftarrow b}$ for the point obtained by setting
coordinate $i$ to $b$.  For $D\subseteq[n]$, define $x^{D\leftarrow1}$
analogously by setting every coordinate in $D$ to one.

Let $\cJ=(J_S)_{S\subseteq[n]}$ be an activation family as in
\eqref{eq:intro-adaptive-map}.  We write
$\cF_{\cJ}=\sigma(\Phi_{\cJ})$ for the $\sigma$-algebra generated by its adaptive
partition.  Thus $\cF_{\cJ}$-measurability is precisely the measurability
notion defined in the introduction.  Following Hatami, an $M$-pseudo-junta is
a Boolean function $h$ that is $\cF_{\cJ}$-measurable for some $\cJ$ satisfying
$\E_p|A_{\cJ}(X)|\le M$.

Our argument uses two quantitative parameters of the activation family.  We
define its \emph{arity}, denoted $\operatorname{arity}(\cJ)$, as the largest $|S|$ for which
$J_S\not\equiv0$, with value zero if no such $S$ exists.  Its
\emph{multiplicity-counted activation load} is $\Lambda_p(\cJ)$ from
\eqref{eq:intro-load}.  These parameters control the active union through
\begin{equation}
\label{eq:union-load}
       |A_{\cJ}(x)|
       \le\sum_{S\subseteq[n]}|S|J_S(x_S),
       \qquad
       \E_p|A_{\cJ}(X)|\le\Lambda_p(\cJ).
\end{equation}
The load can be much larger than the expected active-union size when active
indexing sets overlap.

We use the following explicit extraction from Hatami's $p$-biased construction.

\begin{lemma}[Hatami's construction with arity and load]
\label{lem:hatami-input}
Let $0<p\le1/2$, let $f:\bits^n\to\bits$, and let $0<a\le1$.  Set
$C_f=\lceil I_p(f)\rceil$,
$d=\lceil 10^3C_f/a\rceil$, and
$L=\exp(10^{10}a^{-2}C_f^2)$.
There are increasing activations $\cJ=(J_S)_{S\subseteq[n]}$ and a Boolean
function $h$ such that
\begin{enumerate}[label=\textup{(\roman*)},nosep]
\item $\Pr_p[f\ne h]\le a$;
\item $h$ is $\cF_{\cJ}$-measurable;
\item $\operatorname{arity}(\cJ)\le d$; and
\item $\Lambda_p(\cJ)\le L$.
\end{enumerate}
\end{lemma}

Tracking Hatami's activation family makes its monotonicity, arity, and load
guarantees explicit in the form needed below; see
\cite[Remark~2.9 and Sections~4.3--4.4]{Hatami2012}.  Hatami retains a
bounded-degree family $\cS$ of significant
components in the $p$-biased Walsh expansion
$f(x)=\sum_S F_S(x_S)$.  For
suitable $\eta,\delta>0$, set
$b_T=\sum_{S\in\cS:S\supseteq T}\Pr_p[|F_S(X_S)|\ge\eta]$ and let
$\mu_p^T(y)=p^{|y|}(1-p)^{|T|-|y|}$ be the point mass of $y\in\bits^T$.  The
activation at $T$ is $J_T(y)=\one[b_T\ge\delta\mu_p^T(y)]$.  The degree cutoff
bounds the arity.  When $p\le1/2$, the point mass $\mu_p^T(y)$ is
nonincreasing in $y$, so every $J_T$ is increasing.  Moreover,
\[
       J_T(y)\le\frac{b_T}{\delta\mu_p^T(y)}
       \quad\Longrightarrow\quad
       \E_pJ_T(X_T)\le\frac{2^{|T|}b_T}{\delta},
\]
because $\E_p[1/\mu_p^T(X_T)]=2^{|T|}$.  Summing this estimate with
multiplicity $|T|$ gives the stated load bound.  Appendix~\ref{app:hatami}
gives the details, following \cite[Sections~4.1--4.4]{Hatami2012}.

\section{Controlled monotonicization of adaptive structure}
\label{sec:monotonicization}

\paragraph{Bias-matched shifts.}
Our first ingredient is independent of pseudo-juntas.  Fix a coordinate $i$.
For a Boolean function $v$, write $v_0,v_1$ for its two $i$-slices and define
\[
\begin{aligned}
  (U_i v)_0&=v_0, &\qquad (U_i v)_1&=v_0\vee v_1,\\
  (L_i v)_0&=v_0\wedge v_1, &(L_i v)_1&=v_1.
\end{aligned}
\]
Both shifts fix the nondecreasing fibers $00,01,11$.  On the unique decreasing
fiber $10$, $U_i$ produces $11$ and $L_i$ produces $00$.  OpenAI's
GPT-5.6 Sol suggested trying a random choice between these shifts.

\begin{lemma}[One-coordinate shift]
\label{lem:shift}
Let $0<p<1$, let $i\in[n]$, let $f:\bits^n\to\bits$ be increasing, and let
$v:\bits^n\to\bits$ be arbitrary.  Choose $O_i=U_i$ with probability $1-p$
and $O_i=L_i$ with probability $p$, independently of the input.  Then
\[
       \E_{O_i}\Pr_p[f\ne O_i v]\le\Pr_p[f\ne v].
\]
Moreover, both $U_i v$ and $L_i v$ are nondecreasing in coordinate $i$, and
they preserve monotonicity in every other coordinate in which $v$ is already
nondecreasing.
\end{lemma}

\begin{proof}
Condition on all coordinates other than $i$.  The $i$-fiber of $f$ is one of
$00,01,11$.  Only a $10$ fiber of $v$ changes, and the conditional errors are
\[
\begin{array}{c@{\qquad}c@{\qquad}c}
\toprule
\text{fiber of }f & \text{old error} &
 (1-p)\,\text{error}(11)+p\,\text{error}(00)\\
\midrule
00 & 1-p & 1-p\\
01 & 1   & (1-p)^2+p^2\\
11 & p   & p\\
\bottomrule
\end{array}
\]
The third column never exceeds the second.  Integrating proves the error
bound.  The inequalities $v_0\wedge v_1\le v_1$ and
$v_0\le v_0\vee v_1$ give monotonicity in coordinate $i$; conjunction and
disjunction preserve monotonicity in every other coordinate.
\end{proof}

Starting from $h$, process the coordinates in an arbitrary fixed order,
making the choices in Lemma~\ref{lem:shift} independently.  Let
$D\subseteq[n]$ be the coordinates on which the lower shift was chosen.
Identifying $D$ with its indicator vector, $D\sim\mu_p^n$; every resulting
function $H_D$ is increasing, and iterated conditional expectation gives
\[
       \E_D\Pr_p[f\ne H_D]\le\Pr_p[f\ne h].
\]
The function $H_D$ may depend on the fixed processing order; no commutativity of
the shifts is used.

\paragraph{The adaptive information after shifting.}
For $D\subseteq[n]$, define the forced activation system
$\cJ^D=(J_S^D)_S$ by
\[
       J_S^D(z)=J_S\bigl(z^{D\cap S\leftarrow1}\bigr),
\]
where $z^{D\cap S\leftarrow1}$ is obtained by setting the coordinates in
$D\cap S$ to one.
Its active union is
$A_{\cJ^D}(x)=A_{\cJ}(x^{D\leftarrow1})$.  Crucially, its adaptive atom is
\[
  \Phi_{\cJ^D}(x)
  =\left(A_{\cJ}(x^{D\leftarrow1}),
          x_{A_{\cJ}(x^{D\leftarrow1})}\right).
\]
The active set is computed at the forced input, and the recorded labels are
those of $x$.  When the activations are increasing, the forced partition
refines the original one.  Indeed, every set active at $x$ is contained in
$B=A_{\cJ^D}(x)$, so evaluating $J_S(x_S)$ for $S\subseteq B$ recovers the
original active union from $(B,x_B)$.

\begin{lemma}[Closure under shifts]
\label{lem:closure}
Let $i\in[n]$, let $\cJ=(J_S)_{S\subseteq[n]}$ be an activation family with
increasing activations, and let $v:\bits^n\to\bits$ be
$\cF_{\cJ}$-measurable.  Then $U_i v$ is $\cF_{\cJ}$-measurable and $L_i v$
is $\cF_{\cJ^{\{i\}}}$-measurable.  Consequently, if a sequence of upper and
lower shifts is applied to $v$, and $D$ is the set of coordinates on which a
lower shift occurs, then the final function is $\cF_{\cJ^D}$-measurable.
\end{lemma}

\begin{proof}
Write $v=\psi\circ\Phi_{\cJ}$.  For an atom $(B,a)$ with $i\in B$ and
$a_i=1$, let $a^{i\leftarrow0}$ be obtained by changing the recorded
$i$-label to zero, and put
\[
 \cE_i^-(B,a)=
 \{S\subseteq B:J_S((a^{i\leftarrow0})_S)=1\},\qquad
 B_i^-(B,a)=\bigcup_{S\in\cE_i^-(B,a)}S.
\]
Define
\begin{equation}
\label{eq:R-map}
 R_i(B,a)=
 \bigl(B_i^-(B,a),(a^{i\leftarrow0})_{B_i^-(B,a)}\bigr).
\end{equation}
When $i\notin B$, set $R_i(B,a)=(B,a)$.  Any set active after lowering
coordinate $i$ was already active at the upper endpoint, since $J_S$ is
increasing.  Its indexing set is therefore contained in $B$, and the
recorded labels determine whether it remains active.  If $i\notin B$ at the
upper endpoint, lowering $i$ leaves every activation unchanged.  Hence
$\Phi_{\cJ}(x^{i\leftarrow0})=R_i(\Phi_{\cJ}(x))$ whenever $x_i=1$.

It follows that the upper shift is measurable with respect to the old
partition.  Define
\[
 \psi_U(B,a)=
 \begin{cases}
 \psi(B,a),&i\notin B,\\
 \psi(B,a),&i\in B,\ a_i=0,\\
 \psi(B,a)\vee\psi(R_i(B,a)),&i\in B,\ a_i=1.
 \end{cases}
\]
When $i\in B$, the atom records the endpoint label, and at an upper
endpoint \eqref{eq:R-map} supplies the lower atom.  When $i\notin B$, the
shift leaves the lower endpoint unchanged; at an upper endpoint with
$i\notin B$, lowering $i$ preserves the atom.
Hence $U_iv=\psi_U\circ\Phi_{\cJ}$.

For the lower shift, the upper endpoint can activate new sets containing
$i$.  The forced partition records their union while retaining the actual
labels of the input.  If
$\Phi_{\cJ^{\{i\}}}(x)=(B,a)$, define
\[
 Q_i(B,a)=(B,\widehat a),\qquad
 \widehat a_j=
 \begin{cases}1,&j=i,\\a_j,&j\ne i,\end{cases}
 \quad j\in B.
\]
Then $Q_i(\Phi_{\cJ^{\{i\}}}(x))
=\Phi_{\cJ}(x^{i\leftarrow1})$.  Set
\[
 \psi_L(B,a)=
 \begin{cases}
 \psi(Q_i(B,a)),&i\notin B,\\
 \psi(Q_i(B,a)),&i\in B,\ a_i=1,\\
 \psi(R_i(Q_i(B,a)))\wedge\psi(Q_i(B,a)),
       &i\in B,\ a_i=0.
 \end{cases}
\]
If $i\notin B$, no activation containing $i$ is active even at the upper
endpoint, so the old endpoint atoms coincide.  Otherwise, $Q_i$ gives
the old upper atom and $R_i\circ Q_i$ gives the old lower atom.  Therefore
$L_iv=\psi_L\circ\Phi_{\cJ^{\{i\}}}$.

Forcing preserves monotonicity because $z\mapsto z^{D\cap S\leftarrow1}$
is increasing.  Since $(\cJ^D)^E=\cJ^{D\cup E}$, iterating the
one-coordinate result proves the assertion for a sequence of shifts.
\end{proof}

\begin{remark}[Approximation on a fixed adaptive partition]
\label{rem:adaptive-obstruction}
The following example concerns a general family of increasing activations.
At $p=1/2$, take coordinates $(b,a_1,a_2)$, the increasing target
$f(b,a_1,a_2)=a_1\vee a_2$, and a single nonzero activation
\[
       J_{\{b,a_1,a_2\}}(b,a_1,a_2)=b.
\]
When $b=0$, the four inputs form a single atom.  When $b=1$, each input is
its own atom.  The optimal approximator $h$ labels the $b=0$ atom by one
and agrees with $f$ when $b=1$, giving error $1/8$.  The lower shift
satisfies $L_bh=f$, which takes both values on the $b=0$ atom.  Now consider
an increasing approximator on the original partition.  Labeling the $b=0$
atom by one forces the approximator to be identically one, giving error
$1/4$.  Labeling it by zero gives error at least $3/8$.  Thus requiring the
approximator to be increasing raises the minimum error from $1/8$ to $1/4$.
In Hatami's construction (Section~\ref{sec:pseudo-juntas}), the threshold
$\delta\mu_p^T(y)$ is independent of $y$ at $p=1/2$, so each $J_T$ is constant.
\end{remark}

\paragraph{The cost of forcing.}
We bound the expected size of the forced active union using the arity and
activation load of the original family.

\begin{lemma}[Load transport]
\label{lem:load-transport}
Let $0<p<1$, let $d\in\mathbb Z_{\ge0}$, and let
$\cJ=(J_S)_{S\subseteq[n]}$ be an activation family satisfying
$J_S\equiv0$ whenever $|S|>d$.  Let $D\sim\mu_p^n$ and set $p'=2p-p^2$.
Writing $\Lambda_q$ for the same load evaluated under bias $q$, we have
\begin{equation}
\label{eq:load-transport}
       \E_D\Lambda_p(\cJ^D)
       =\Lambda_{p'}(\cJ)
       \le(2-p)^d\Lambda_p(\cJ)
       \le2^d\Lambda_p(\cJ).
\end{equation}
Consequently,
\[
       \E_D\E_p|A_{\cJ^D}(X)|
       \le(2-p)^d\Lambda_p(\cJ).
\]
\end{lemma}

\begin{proof}
Let $X\sim\mu_p^n$ be independent of $D$.  Each coordinate $(X^{D\leftarrow1})_i=X_i\vee\one_{\{i\in D\}}$
is one with probability $p+(1-p)p=p'$, and these coordinates are independent.
Thus $X^{D\leftarrow1}$ has law $\mu_{p'}^n$.  In particular,
$\E_D\E_p|A_{\cJ^D}(X)|=\E_{p'}|A_{\cJ}(X)|$.
Applying the same change of measure to each activation gives the equality
in \eqref{eq:load-transport}.

Let $|S|=s\le d$ and let $z\in\bits^S$ have Hamming weight $r$.  Since
$p'/p=2-p$ and $(1-p')/(1-p)=1-p$,
\[
  \frac{\mu_{p'}^S(z)}{\mu_p^S(z)}
  =(2-p)^r(1-p)^{s-r}
  \le(2-p)^s\le(2-p)^d.
\]
Integrating the nonnegative function $J_S$, multiplying by $|S|$, and summing
over $S$ proves the inequalities.  The final claim follows from
\eqref{eq:union-load}.
\end{proof}

We can now state and prove the controlled monotonicization theorem used in the
main argument.

\begin{theorem}[Controlled monotonicization]
\label{thm:controlled-monotonicization}
Let $0<p<1$, let $d\in\mathbb Z_{\ge0}$, let $e,L\ge0$, and let
$f:\bits^n\to\bits$ be increasing.  Let
$\cJ=(J_S)_{S\subseteq[n]}$ be an activation family with increasing
activations, and let $h:\bits^n\to\bits$ be $\cF_{\cJ}$-measurable.  Suppose
$J_S\equiv0$ for $|S|>d$ and
\[
       \Pr_p[f\ne h]\le e,
       \qquad
       \Lambda_p(\cJ)\le L.
\]
For every $\tau>0$, there exist $D\subseteq[n]$ and an increasing Boolean
function $u$ such that
\begin{enumerate}[label=\textup{(\roman*)}]
\item $u$ is $\cF_{\cJ^D}$-measurable;
\item $\Pr_p[f\ne u]\le e+\tau$; and
\item $\E_p|A_{\cJ^D}(X)|\le(2-p)^d L(1+e/\tau)$.
\end{enumerate}
\end{theorem}

\begin{proof}
Apply the random shifts in a fixed coordinate order.  For each realization $D$, let
$e_D=\Pr_p[f\ne H_D]$ and $M_D=\E_p|A_{\cJ^D}(X)|$.
By the shift construction, Lemma~\ref{lem:closure}, and
Lemma~\ref{lem:load-transport}, every $H_D$ is increasing and
$\cF_{\cJ^D}$-measurable, while
\[
       \E_D e_D\le e,
       \qquad
       \E_D M_D\le B:=(2-p)^d L.
\]
If $B=0$, then nonnegativity gives $M_D=0$ for every realization with
positive probability; choose one with $e_D\le e$.  Otherwise,
\[
       \E_D\left[e_D+\frac{\tau}{B}M_D\right]\le e+\tau.
\]
Choose one $D$ for which the expression inside brackets is at most its
expectation.  Nonnegativity gives
$e_D\le e+\tau$ and
$M_D\le B(1+e/\tau)$.  Taking $u=H_D$ proves the theorem.
\end{proof}

\section{Certificates and completion of the proof}
\label{sec:certificates}

When the approximator $u$ and its activations are increasing, the active
coordinates at an accepted input determine a positive certificate for $u$.

\begin{lemma}[The active $1$-coordinates form a positive certificate]
\label{lem:certificates}
Let $\cJ=(J_S)_{S\subseteq[n]}$ be an activation family with increasing
activations, and let $u:\bits^n\to\bits$ be increasing and
$\cF_{\cJ}$-measurable.  For $x\in\bits^n$, set
$T(x)=\{i\in A_{\cJ}(x):x_i=1\}$.  If $u(x)=1$, then $T(x)$ is a positive certificate for $u$: every
$y\in\bits^n$ satisfying $y_i=1$ for all $i\in T(x)$ has $u(y)=1$.
Consequently, for every integer $w\ge0$,
\begin{equation}
\label{eq:certificate-dnf}
  g_w(y)=
  \bigvee_{\substack{x\in\bits^n:\ u(x)=1\\ |T(x)|\le w}}
  \ \bigwedge_{i\in T(x)}y_i
\end{equation}
is a monotone DNF of width at most $w$, satisfies $g_w\le u$, and obeys
\begin{equation}
\label{eq:certificate-error}
  \Pr_p[u\ne g_w]
  \le\Pr_p[|T(X)|>w]
  \le\frac{\E_p|A_{\cJ}(X)|}{w+1}.
\end{equation}
\end{lemma}

\begin{proof}
Fix $x$ with $u(x)=1$ and let $x^\circ$ be the incidence vector of $T(x)$.
Equivalently, $x^\circ$ agrees with $x$ on $A_{\cJ}(x)$ and is zero outside
that set.  We claim that every activation has the same value at $x^\circ$ and
$x$.  If $J_S(x_S)=1$, then $S\subseteq A_{\cJ}(x)$, so
$x_S^\circ=x_S$.  If $J_S(x_S)=0$, then $x^\circ\le x$ and monotonicity gives
$J_S(x_S^\circ)=0$.  Hence
$\Phi_{\cJ}(x^\circ)=\Phi_{\cJ}(x)$, and
$\cF_{\cJ}$-measurability implies $u(x^\circ)=u(x)=1$.

If $y_i=1$ for every $i\in T(x)$, then $y\ge x^\circ$, so monotonicity of $u$
gives $u(y)=1$.  Thus every term in \eqref{eq:certificate-dnf} is an implicant
of $u$, and every accepted $x$ with $|T(x)|\le w$ satisfies its own term.
Therefore
$\{x:u(x)\ne g_w(x)\}\subseteq\{x:|T(x)|>w\}$.  Since $|T(X)|$ is
integer-valued, Markov's inequality at threshold $w+1$, together with
$T(x)\subseteq A_{\cJ}(x)$, gives \eqref{eq:certificate-error}.
\end{proof}

\begin{proof}[Proof of Theorem~\ref{thm:main}]
If $f$ is constant, use the corresponding width-zero DNF\@.  Otherwise, put $C_K=\lceil K\rceil$ and $a=\varepsilon/4$.
Since $\lceil I_p(f)\rceil\le C_K$, Lemma~\ref{lem:hatami-input} gives an
approximator $h$ and increasing
activations $\cJ$ such that
\[
  \Pr_p[f\ne h]\le a,
  \qquad
  h\text{ is }\cF_{\cJ}\text{-measurable},
  \qquad
  \operatorname{arity}(\cJ)\le d,
  \qquad
  \Lambda_p(\cJ)\le L,
\]
where
\[
  d=\left\lceil\frac{4000C_K}{\varepsilon}\right\rceil,
  \qquad
  L=\exp\!\left(\frac{16\cdot10^{10}C_K^2}{\varepsilon^2}\right).
\]
Apply Theorem~\ref{thm:controlled-monotonicization} with $e=\tau=a$.  Since
$(2-p)^d\le2^d$, it gives $D\subseteq[n]$ and an increasing
$\cF_{\cJ^D}$-measurable function $u$, where the forced activations remain
increasing and
\[
       \Pr_p[f\ne u]\le\frac{\varepsilon}{2},
       \qquad
       \E_p|A_{\cJ^D}(X)|\le2^{d+1}L.
\]

Choose $w=\lceil 4\cdot 2^d L/\varepsilon\rceil$.  Applying Lemma~\ref{lem:certificates} to $u$ gives a monotone width-$w$ DNF
$g_w$ with
\[
\begin{aligned}
       \Pr_p[f\ne g_w]
       &\le\Pr_p[f\ne u]+\Pr_p[u\ne g_w]\\
       &\le\frac{\varepsilon}{2}
          +\frac{\E_p|A_{\cJ^D}(X)|}{w+1}
       \le\frac{\varepsilon}{2}
          +\frac{2^{d+1}L}{w+1}
       \le\varepsilon.
\end{aligned}
\]
Finally, $C_K\le K+1$ and $0<\varepsilon<1$, so
\[
  \log w
  =O\!\left(
       \log\frac1\varepsilon
       +\frac{K+1}{\varepsilon}
       +\frac{(K+1)^2}{\varepsilon^2}
       \right)
  =O\!\left(\frac{(K+1)^2}{\varepsilon^2}\right).
\]
This proves the theorem.
\end{proof}

\begin{proof}[Proof of Conjecture~\ref{conj:friedgut}]
Let $f=\one_{\cA}$.  By \eqref{eq:intro-margulis-russo},
\begin{equation}
\label{eq:friedgut-normalization}
       I_p(f)=2(1-p)\,p\E_p[\sens_f(X)].
\end{equation}
If $p\le1/2$, then $I_p(f)\le2c$, and
Theorem~\ref{thm:main} applies.

Suppose $p>1/2$.  The product Poincar\'e inequality and
\eqref{eq:friedgut-normalization} give
\[
       \Var_p(f)\le\frac12 I_p(f)
       \le c(1-p).
\]
Let $\alpha=\E_p f$ and $m=\min\{\alpha,1-\alpha\}$.  Since
$\Var_p(f)=m(1-m)\ge m/2$, if
$1-p\le\varepsilon/(2c)$ then $m\le\varepsilon$, and the closer constant
function is an admissible width-zero approximation.

In the remaining case, $\lambda=\varepsilon/(2c)<1-p<1/2$, so
$p\in[\lambda,1-\lambda]$.  The relative boundary assumption directly
bounds the total influence in the bit-flip convention:
\[
  \sum_{i=1}^n\Pr_p[f(X)\ne f(X^{\oplus i})]
  =\E_p[\sens_f(X)]\le c/p<2c.
\]
On this interval, Friedgut's biased junta theorem gives a junta $h$ with
$\Pr_p[f\ne h]\le\varepsilon$ whose coordinate set $J$ has size bounded
in terms of $c$ and $\varepsilon$ alone \cite{Friedgut1998,ODonnell2014}.
The conditional mean $m_J(z)=\E_p[f(X)\mid X_J=z]$ is increasing, and
its Bayes threshold
$\one[m_J(z)\ge1/2]$ has no larger error than any other $J$-measurable Boolean
function.  Every increasing $J$-junta has a monotone DNF of width at most
$|J|$.  Taking the maximum of this bound and the low-bias bound proves the
conjecture.
\end{proof}

\begin{remark}[Resampling influence as $p\to1$]
\label{rem:high-bias-and}
A bound on $I_p(f)$ alone cannot give a width bound uniform as $p\to1$.  Let
$f(x)=x_1\wedge\cdots\wedge x_n$ and $p=1-1/n$.  Then
\[
       \Pr_p[f=1]=p^n\longrightarrow e^{-1},
       \qquad
       I_p(f)=2p(1-p)np^{n-1}\longrightarrow\frac2e.
\]
For fixed $w$, every nonempty monotone DNF of width at most $w$ accepts with
probability at least $p^w\to1$, while the empty DNF is identically zero.
Hence no fixed-width monotone DNF approximates $f$ to sufficiently small
constant error.  Friedgut's all-bias quantity is instead
$p\E_p[\sens_f(X)]=np^n$, which diverges in this example.
\end{remark}

The proof gives an existential bound on width.  The DNF satisfies $g_w\le u$,
and its guarantee for the original function $f$ is the two-sided error bound
$\Pr_p[f\ne g_w]\le\varepsilon$.  The optimal dependence of the width on $K$
and $\varepsilon$ remains open.  It would also be of interest to obtain
bounds on the number of terms and the computational cost of finding such a DNF.

\section*{AI methodology statement}
We used OpenAI's GPT-5.6 Sol to work through the relevant literature and
explore possible proof strategies.  Its main contribution was to suggest
using randomized coordinate shifts to make Hatami's approximator increasing
(see Lemma~\ref{lem:shift}).  With assistance from OpenAI's GPT-6 Astra,
we revised the exposition and completed a
\href{https://github.com/liuchliuch/narrow-dnf-lean}{\pinkfoldmark\nobreak\hspace{0.2em}Lean formalization}
of our proofs of Theorem~\ref{thm:main} and Friedgut's conjecture.
We independently checked and finalized all claims and proofs and take
responsibility for the paper.

\appendix

\section{Hatami's construction with arity and load}
\label{app:hatami}

We prove Lemma~\ref{lem:hatami-input} by tracking the activation family in
Hatami's Fourier-analytic proof.  We verify its arity and
multiplicity-counted activation load bounds while retaining Hatami's
approximation guarantee.  All
probabilities, expectations, and $L^2$-norms below are taken under $\mu_p^n$
or the appropriate coordinate marginal.

Let $f(x)=\sum_{S\subseteq[n]}F_S(x_S)$ be the $p$-biased Walsh expansion.  In Hatami's notation,
$F_S(x_S)=\widehat f(S)\prod_{i\in S}r_p(x_i)$, where
$\widehat f(S)$ is the $p$-biased Fourier coefficient and
\[
       r_p(0)=-\sqrt{\frac{p}{1-p}},
       \qquad
       r_p(1)=\sqrt{\frac{1-p}{p}}.
\]
Parseval's identity and the resampling definition of influence give
\[
       \sum_S\|F_S\|_2^2=\|f\|_2^2\le1,
       \qquad
       I_p(f)=2\sum_S |S|\|F_S\|_2^2.
\]
If $I_p(f)=0$, full support and connectivity of the cube imply that $f$ is
constant; take $h=f$ and all activations zero.  Assume henceforth that
$I_p(f)>0$, and put $C_f=\lceil I_p(f)\rceil$.  Following
\cite[Section~4]{Hatami2012}, set
\[
  \varepsilon_0=10^{-3}a,
  \qquad
  k=C_f\varepsilon_0^{-1}=10^3a^{-1}C_f,
  \qquad
  \delta=2^{-10^2k^2},
  \qquad
  \varepsilon_1=3^{-10k^2}\varepsilon_0^{10k}.
\]
Retain the Fourier sets
\[
       \cS=\{S\subseteq[n]: |S|\le k,
                    \ \|F_S\|_\infty>\varepsilon_1\}.
\]
For $T\subseteq[n]$, define
\[
       b_T=\sum_{\substack{S\in\cS\\S\supseteq T}}
       \Pr_p\bigl[|F_S(X_S)|\ge\varepsilon_1\bigr].
\]
For $|T|\le k$ and $y\in\bits^T$, set
\begin{equation}
\label{eq:hatami-activation}
       J_T(y)=\one[b_T\ge\delta\mu_p^T(y)],
\end{equation}
and put $J_T\equiv0$ for $|T|>k$.  This is precisely the activation family in
\cite[Section~4.3]{Hatami2012}.  The number $b_T$ is global and independent of
$y$.  Since $p\le1/2$, the point mass $\mu_p^T(y)$ is nonincreasing in $y$, so
every $J_T$ is increasing.  Also, nonzero activations have arity at most
$k\le d$.

It remains to record the load estimate in the form needed here.  From
\eqref{eq:hatami-activation}, pointwise in $y$,
\[
       J_T(y)\le\frac{b_T}{\delta\mu_p^T(y)}.
\]
If $X_T\sim\mu_p^T$, then full support gives
\[
       \E_p\left[\frac1{\mu_p^T(X_T)}\right]
       =\sum_{y\in\bits^T}1=2^{|T|}.
\]
Consequently,
\begin{align*}
 \Lambda_p(\cJ)
 &=\sum_T |T|\E_pJ_T(X_T)\\*
 &\le \frac{k}{\delta}
      \sum_{\substack{T\subseteq[n]\\|T|\le k}}2^{|T|}
      \sum_{\substack{S\in\cS\\S\supseteq T}}
      \Pr_p[|F_S(X_S)|\ge\varepsilon_1]\\*
 &=\frac{k}{\delta}
      \sum_{S\in\cS}\Pr_p[|F_S(X_S)|\ge\varepsilon_1]
      \sum_{T\subseteq S}2^{|T|}\\*
 &\le\frac{k2^{2k}}{\delta\varepsilon_1^2}
      \sum_{S\in\cS}\|F_S\|_2^2
 \le\frac{k2^{2k}}{\delta\varepsilon_1^2}.
\end{align*}
Here $\sum_{T\subseteq S}2^{|T|}=3^{|S|}\le2^{2k}$, and the penultimate
inequality is Markov's inequality.  For completeness, the last parameter estimate follows directly from the
chosen constants.  Since $I_p(f)>0$, we have $C_f\ge1$, hence
$k=C_f/\varepsilon_0\ge10^3$ and $\varepsilon_0^{-1}\le k$.  Therefore
\[
 \frac{k2^{2k}}{\delta\varepsilon_1^2}
 =k2^{2k+10^2k^2}3^{20k^2}\varepsilon_0^{-20k}
 \le2^{10^3k^2}
 \le\exp\!\left(10^{10}a^{-2}C_f^2\right)=L.
\]
This is the displayed estimate on pages~518--519 of \cite{Hatami2012}.
Thus the same construction has the required multiplicity-counted activation
load.

Finally, for exactly the activations \eqref{eq:hatami-activation}, combining
the estimates from Steps I--III of Hatami's proof with the final calculation in
Section~4.4 yields
\[
       \left\|f-\E_p[f\mid\cF_{\cJ}]\right\|_2^2
       \le10\varepsilon_0;
\]
see equations (10)--(16) and Section~4.4 of \cite{Hatami2012}.  Define
\[
       h(x)=\one\!\left[
          \E_p[f\mid\cF_{\cJ}](x)>\frac12
       \right].
\]
Then $h$ is $\cF_{\cJ}$-measurable.  Since $f$ is Boolean,
\[
       \Pr_p[f\ne h]
       \le4\left\|f-\E_p[f\mid\cF_{\cJ}]\right\|_2^2
       \le40\varepsilon_0\le a.
\]
This proves Lemma~\ref{lem:hatami-input}.

\end{document}